\documentclass[aps,pra,twocolumn,superscriptaddress]{revtex4-1}

\usepackage[english]{babel}
\usepackage{amsmath}
\usepackage{amsthm}
\usepackage{amssymb}
\usepackage{mathrsfs}
\usepackage{booktabs}
\usepackage{color}
\usepackage{hyperref}
\usepackage[ruled,vlined]{algorithm2e}
\usepackage{cases}
\usepackage{dblfloatfix}
\usepackage{multirow}
\hypersetup{
  colorlinks   = true,  
  urlcolor     = blue,  
  linkcolor    = blue,  
  citecolor    = red    
}
\usepackage{mathtools}
\usepackage{natbib}
\usepackage{pgfplots}
\usepackage{subfigure}
\usepackage{url}

\newcommand{\rvline}{\hspace*{-\arraycolsep}\vline\hspace*{-\arraycolsep}}

\newcommand{\di}{{\rm d}}
\newcommand{\C}{{\sf C}}
\newcommand{\F}{{\sf F}}
\newcommand{\U}{{\sf U}}
\newcommand{\V}{{\sf V}}
\newcommand{\Y}{{\sf Y}}
\renewcommand{\L}{{\sf L}}
\renewcommand{\S}{{\sf S}}
\newcommand{\R}{{\sf R}}
\newcommand{\Z}{{\sf Z}}

\newcommand{\x}{\boldsymbol{x}}

\mathtoolsset{showonlyrefs}

\theoremstyle{definition}

\theoremstyle{remark}

\theoremstyle{plain}
\newtheorem*{theorem*}{Theorem}
\newtheorem*{lemma}{Lemma}

\usepackage{scalerel,stackengine}
\stackMath
\newcommand\reallywidehat[1]{%
\savestack{\tmpbox}{\stretchto{%
  \scaleto{%
    \scalerel*[\wi\di thof{\ensuremath{#1}}]{\kern-.6pt\bigwedge\kern-.6pt}%
    {\rule[-\textheight/2]{1ex}{\textheight}}
  }{\textheight}%
}{0.5ex}}%
\stackon[1pt]{#1}{\tmpbox}%
}
\usepackage{comment}
\usepackage{float}

\SetArgSty{textnormal}



\begin{document}

\title{Qubit fidelity under stochastic Schr\"{o}dinger equations driven by colored noise}

\author{R.J.P.T. \surname{de Keijzer}}
\affiliation{Department of Applied Physics, Eindhoven University of Technology, P. O. Box 513, 5600 MB Eindhoven, The Netherlands}
\affiliation{Eindhoven Hendrik Casimir Institute, Eindhoven University of Technology, P. O. Box 513, 5600 MB Eindhoven, The Netherlands}
\altaffiliation[Corresponding author: ]{r.j.p.t.d.keijzer@tue.nl }

\author{L.Y. \surname{Visser}}
\author{O. \surname{Tse}}
\affiliation{Department of Mathematics and Computer Science, Eindhoven University of Technology, P.~O.~Box 513, 5600 MB Eindhoven, The Netherlands}
\affiliation{Eindhoven Hendrik Casimir Institute, Eindhoven University of Technology, P. O. Box 513, 5600 MB Eindhoven, The Netherlands}

\author{S.J.J.M.F. \surname{Kokkelmans}}
\affiliation{Department of Applied Physics, Eindhoven University of Technology, P. O. Box 513, 5600 MB Eindhoven, The Netherlands}
\affiliation{Eindhoven Hendrik Casimir Institute, Eindhoven University of Technology, P. O. Box 513, 5600 MB Eindhoven, The Netherlands}

\date{\today}

\begin{abstract}
Environmental noise on a controlled quantum system is generally modeled by a dissipative Lindblad equation. This equation describes the average state of the system via the density matrix $\rho.$ One way of deriving this Lindblad equation is by introducing a  stochastic operator evolving under white noise in the Schr\"{o}dinger equation. However, white noise, where all noise frequencies contribute equally in the power spectral density, is not a realistic noise profile as lower frequencies generally dominate the spectrum. Furthermore, the Lindblad equation does not fully describe the system as a density matrix $\rho$ does not uniquely describe a probabilistic ensemble of pure states $\{\psi_j\}_j$. In this work, we introduce a method for solving for the \textit{full distribution} of qubit fidelity driven by important stochastic Schr\"{o}dinger equation cases, where qubits evolve under more realistic noise profiles, e.g.\ \textit{Ornstein-Uhlenbeck noise}. This allows for predictions of the mean, variance, and higher-order moments of the fidelities of these qubits, which can be of value when deciding on the allowed noise levels for future quantum computing systems, e.g.\ deciding what quality of control systems to procure. Furthermore, these methods will prove to be integral in the \textit{optimal control} of qubit states under (classical) control system noise. 
\end{abstract}

\maketitle

\section{Introduction}
\label{sec:introduction}

For quantum computers to become important tools in highly complex simulation/optimization problems, they should be able to perform many successive operations faithfully, i.e.\ be able to manipulate the qubits with high fidelity. In the current noisy intermediate-scale quantum (NISQ) era \cite{Preskill_2018}, quantum computers suffer from noise stemming from both uncontrollable sources such as radiative decay, and from sources in control systems, such as frequency and intensity noises \cite{psdlimit,noise2,noise3}. One way of addressing the latter is by designing control operations that are robust against these types of noise \cite{madhav,sven}. Alternatively, these controllable noise sources can be diminished by increasing the quality of the systems, which often comes at both monetary and technological complexity costs. It is therefore important to know in what way noise is determinantal to the fidelity of the qubit operations. A model that predicts fidelities based on noise spectra can be used to ensure that a system meets certain criteria for qubit fidelities, before being acquired. Furthermore, such a model can be used to \textit{control} qubit systems \cite{vqoc,meitei,noiselindblad}, where the noise affecting each state preparation is accounted for.

\smallskip

Commonly, noisy quantum systems can be regarded as open quantum systems, and are often modeled by the Lindblad master equation \cite{lindblad1}. These equations describe the effective Markovian behavior of a quantum state and can be derived by various approaches \cite{lindblad3} (cf.\ Sec.~\ref{sec:SSE} for an example). However, Markovianity is in general an unrealistic assumption for a quantum system subjected to control noise, where lower frequencies often dominate the power spectral density (PSD) \cite{psdwelch}. Moreover, a density matrix does not completely characterize the state of the system, as two ensembles of states $\{\psi_j\}_j$ might share a common density matrix $\rho$. For example, consider two ensembles $\{\{\psi_1=|0\rangle,\ p_1=1/2\},\{\psi_2=|1\rangle,\ p_2=1/2\}\}$ and $\{\{\psi_1=|+\rangle,\ p_1=1/2\},\{\psi_2=|-\rangle,\ p_2=1/2\}\}$ (where $|+\rangle$ and $|-\rangle$ are the eigenstates of the Pauli $X$ operator, then these would be described by the same density matrix $\rho$ but would behave differently in subsequent evolution. Therefore, one might not only be interested in the effective (or average) state of the system but also in the variances of the evolution, which indicate the variety of states being prepared.

\begin{figure}[b]
    \centering
    \includegraphics[scale=0.43]{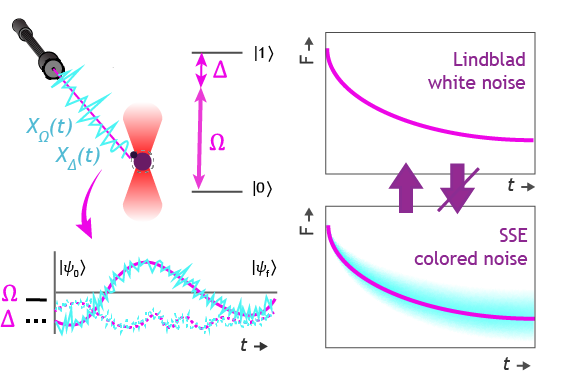}
    \caption{Illustration of control system laser noise $X^\Omega$ and $X^\Delta$ respectively acting on qubit Rabi frequency $\Omega$ and detuning $\Delta$ causing loss of fidelity $\F$. The SSE, which implies the Lindblad equation for white noise, can be used to treat colored noise and to find a distribution of fidelities. Conversely, the Lindblad equation method is only valid for white noise and results in an average fidelity rather than a full distribution.}
    \label{fig:enter-label}
\end{figure}

\smallskip

Because of these shortcomings in the Lindblad equation, we instead consider the \textit{stochastic Schr\"{o}dinger equation} (SSE) as a starting point for our analysis, allowing us to model the evolution of a quantum state under more general noise profiles. Using this model, we derive stochastic ODE systems for the qubit fidelity distribution ${\sf F}$ of qubit states under (classical) control system noise that we verify numerically with Monte-Carlo methods. These ODE systems are substantially simpler and less computationally demanding than their Monte Carlo counterparts.

\smallskip

In line with this work, other studies have been performed to understand the influence of noise on qubit fidelity \cite{saffmannwhitenoise,psdlimit,Green_2013}, where the Lindblad equation is considered primarily in the small noise regime, without providing the full qubit fidelity distribution. Similarly, past works on the SSE for quantum computing \cite{sse1,sse2, boundspaper,boundspaper2,psdpaper} only provide lower bounds or approximations on the achievable fidelities of qubit state preparations. As only white noise and mean fidelities are considered, these calculations could also have been performed using the Lindblad equation. To go beyond the Lindblad equation, these SSE methods are expanded in this work to obtain exact results for more realistic noise acting on both single and multiple qubits. For this, we extend the colored noise methods for the SSE developed in \cite{colorednoisepaper} to the class of {\em It\^o processes} having finite quadratic variation \cite{semimartingale,wong1971representation}.

\smallskip

The layout of this paper is as follows. Sec.~\ref{sec:SSE} describes the SSE, its relation to noise induced by the control mechanism, and the Lindblad equation in the case of white noise.  Sec.~\ref{sec:method} details our ODE approach and the numerical scheme used to solve the SSE equations. In Sec.~\ref{sec:results}, we show initial results for our model on single and multiple qubits, using white and colored noise.

\section{Stochastic Schr\"{o}dinger Equation}
\label{sec:SSE}

The standard Schr\"{o}dinger equation describing the evolution of a pure state $\phi$ is given by \footnote{In this work, we choose to refrain from using standard bra-ket notation for readability purposes.}
\begin{equation}\label{eq:se}
    \di \phi=-iH\phi\, \di t,
\end{equation}
where $H=H(t)$ is a time-dependent Hamiltonian, which is commonly the sum of a drift Hamiltonian and a control Hamiltonian. Most single qubit systems can be controlled using a \emph{coupling operator} 
\[
	H_{coup} = e^{i\varphi}|0\rangle\langle1|+e^{-i\varphi}|1\rangle\langle0|,
\]
and a \emph{detuning operator}
\[
	H_{det} = (I-\sigma_Z)/2=|1\rangle\langle1|,
\]
which in combination gives the control Hamiltonian \cite{vqoc}
\begin{equation}
\label{eq:hamiltonian}
    H(t)=\Omega(t) H_{coup}(t) +\frac{1}{2}\Delta(t)H_{det},
\end{equation}
where $\Omega$ is the Rabi frequency, $\varphi$ is the phase and $\Delta$ is the detuning. In for instance neutral atom systems, these are respectively related to the intensity, phase, and frequency of the control lasers \cite{noiseneutralatom}. The coupling and detuning will always include noise in both the intensity, phase, and frequency, such that the effective coupling, phase, and detuning are given by $\tilde{\Omega}=\Omega+\di X^\Omega$, $\tilde{\varphi}=\varphi+\di X^\varphi$  and $\tilde{\Delta}=\Delta+\di X^\Delta$, where $X^\eta=(X_t^\eta)_{t\ge 0}$, $\eta\in\{\Omega,\varphi,\Delta\}$ are the classical noise sources.

\smallskip

We consider noise sources $X=(X_t)_{t\ge 0}$ that are It\^o processes with finite quadratic variation \cite{ito,quadvar} given by
\[
	[X]_t = \int_0^t \gamma_s^2\,ds.
\]
Here, $\gamma$ is a predictable process that is integrable w.r.t.\ the Brownian motion. In this work, we will only consider processes with $[X]_t = \gamma^2 t$, for some constant $\gamma>0$. Including a noise source in the standard Sch\"{o}dinger equation results in the SSE for the state $\psi$ \cite{Barchielli_2010,colorednoisepaper} (see App.~\ref{app:ssederivation} for derivation)
\begin{equation}\label{eq:sse}
	\di \psi=-iH\psi\, \di t-\frac{1}{2} S^2 \di [X]_t - i S\psi\, \di X_t,\tag{SSE}
\end{equation}
where $S=S^\dag$ is a Hermitian noise operator. 

\smallskip

In this work, we distinguish between three classes of noise operators $S$ in \eqref{eq:sse},
\begin{align*}
    \text{Pauli noise:}&\quad [H,S]=0,\quad S^2=I\\ 
    &\quad\text{(e.g. } S=\sigma_X, \sigma_Y, \sigma_Z \text{),}\\[0.5em]
    \text{Projection noise:}&\quad [H,S]=0, \quad S^2=S\\
    &\quad\text{(e.g. } S=(I-\sigma_Z)/2\text{).}\\[0.5em]
    \text{Non-commuting:}&\quad H=\alpha\sigma_1,\quad S=\sigma_2, \\
    &\quad [H,S]=2i\alpha\sigma_3,
\end{align*}

Thus, following \eqref{eq:hamiltonian}, Pauli noise would model noise in the intensity of the control (or phase via a formal expansion of the exponential function), and projection noise would model noise in the frequency of the control. For the non-commuting case, $\{\sigma_1,\sigma_2,\sigma_3\}$ could for example be any cyclic permutation of $\{\sigma_X,\sigma_Y,\sigma_Z\}$, and can be used to model noise when the detuning and Rabi frequency are of the same order of magnitude.

\smallskip

For our simulation results, we consider the following two examples of noise processes:

\smallskip

\paragraph*{White noise (WN) process.} Here, $X$ takes the form
\begin{equation}
	X_t = \gamma \int_0^t \di W_s,\qquad \gamma\geq 0,
\end{equation}
where $W=(W_t)_{t\ge 0}$ is the standard Brownian motion, and has the quadratic variation $[X]_t=\gamma\, \di t$. Hence, the corresponding SSE reads
\begin{equation}\label{eq:sse-wn}
	\di \psi=-iH\psi \,\di t-\frac{1}{2}\gamma^2 S^2\, \di t - i\gamma S\psi\, \di W_t\tag{WN}.
\end{equation}
Using \eqref{eq:sse-wn} to deduce the evolution of the state $\psi\psi^\dagger$ and taking the expectation results in the Lindblad equation  \cite{lindblad1,lindblad2} (see App.~\ref{app:ssederivation})
\begin{equation}
\label{eq:lindblad}
    \partial_t \rho = -i \left[H, \rho \right] + S \rho S^\dag - \frac{1}{2}\bigl\{S^\dag S, \rho \bigr\},
\end{equation}
where $\rho=\mathbb{E}[\psi\psi^\dag]$.

\smallskip

\paragraph*{Ornstein-Uhlenbeck (OU) process.} For any constant $k>0$, the (OU) process is given by
\begin{equation}
	\di X_t=-kX_t\,\di t+\gamma\,\di W_t,
\end{equation}
subjected to either the calibrated initial data, $X_0=0$, or distributed according to the stationary distribution of the (OU) process, i.e.\ $X_0\sim \gamma\mathcal{N}/\sqrt{2k}$, where $\mathcal{N}$ is the standard normal distribution. The (OU) process can be seen as a damped (WN) process according to the fluctuation-dissipation relation, or conversely, (WN) is (OU) with $k\downarrow0$ \cite{oufluc}. The SSE with (OU) noise reads
\begin{equation}\label{eq:sse-ou}
	\di \psi=-iH\psi \,\di t - \frac{1}{2}\gamma^2 S^2\psi \,\di t - i S\psi\,\di X_t.
\end{equation}
Note that it has the same quadratic variation as (WN).

\begin{figure}
    \centering
    \includegraphics[scale=0.45]{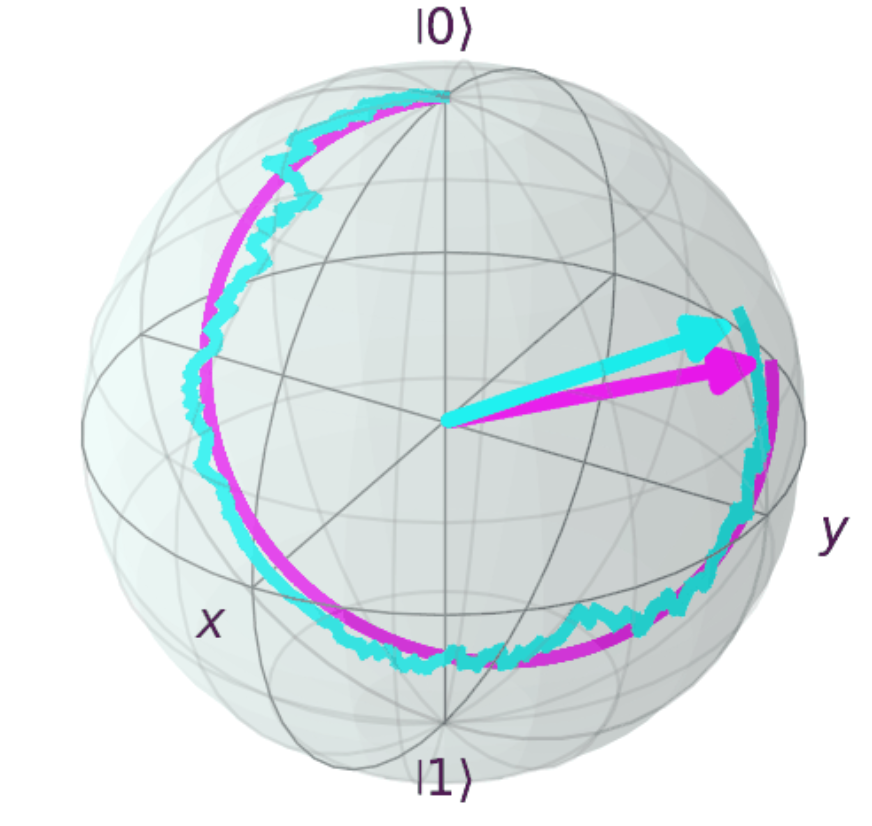}
    \caption{Example of trajectories of state evolutions on the Bloch sphere with $H=X+Z$ without noise (pink), and with Pauli (OU) noise with $S=X, k=0.1, \gamma=0.2$ (cyan).}
    \label{fig:noisyblochsphere}
\end{figure}

\section{Analytical methods}
\label{sec:method}
In this section, we describe in detail the methods used to solve stochastic Schr\"{o}dinger equations for qubit systems. Since the numerical integration of the SSE for $\psi$ can be computationally expensive and unstable due to the non-Euclidean state space of the qubit and the possible non-linearity introduced by a general noise process, we derive full distribution solutions, which in turn lead to explicit formulas for the expectation and variance instead of approximations and/or bounds. 

\subsection{System of ODEs}
\label{sec:odes}
Similar to \cite{boundspaper}, the quantity of interest in quantum computing is the fidelity $\F=|\phi^\dagger \psi|^2$, where $\phi$ is the desired state without noise, i.e.\ $\phi$ satisfies \eqref{eq:se}, and $\psi$ is the state evolving under noise, i.e.\ $\psi$ satisfies \eqref{eq:sse}.

\smallskip

To derive an explicit formula for $\F$, we recursively derive a system of real-valued stochastic differential equations for a vector $\V\in \mathbb{R}^m$, $m\ge 1$, as
\begin{equation}
\label{eq:odesystem}   
\di\V=A\V\, \di t+B\V\, \di X_t+a\,\di t+b\,\di X_t,
\end{equation}
where the first component of $\V$ is the fidelity $\F=|\phi^\dagger\psi|^2$. The size of the system varies depending on the properties of the noise operator $S$. In the case of Pauli and projection noise, $m=3$, while for the non-commuting case, $m=10$.

\smallskip

For white noise (WN), we can solve these equations by considering $\Y=\mathbb{E}[\V]$ and solving for the deterministic equation $\di\Y=A\Y\,\di t+a\,\di t$, which is equivalent to the Lindbladian approach (see App.~\ref{app:ssederivation}). However, for a general noise (SM), this can no longer be done because $\mathbb{E}[\V\, \di X_t]\neq\mathbb{E}[\V]\mathbb{E}[\di X_t]$. Specifically, for (OU), we see terms of the form $\mathbb{E}[X^2\V]$ appearing. The physically most obvious way of incurring a solution, as performed in \cite{orderofapprox}, is to extend the vector $\V$ with new entries of the form $X^{2n}\V$, and derive an elaborate system which will then again include $\mathbb{E}[X^{2n+1}\V]$. Eventually, an approximation has to be made for these terms to keep the system finite. In App.~\ref{app:itoiso}, we detail on the approximation $\mathbb{E}[X^{2n}\V]=\mathbb{E}[X^{2n}]\mathbb{E}[\V]$. The $n$ at which the truncation is done is referred to as the \textit{order} of the approximation, for which an analytical solution for $\mathbb{E}[\V]$ can be found. As shown in Sec.~\ref{sec:resultsa}, solutions become better with increasing order of approximation, but eventually become nonphysical as fidelities will leave the interval $[0,1]$.

\smallskip

Instead, we argue for the less intuitive approach of solving for the full distribution of $\V$, from which the expectation and higher-order moments can be inferred. For all cases considered in this work---except for the non-commuting case---we obtain a diagonalizable ODE system \eqref{eq:odesystem}, allowing us to obtain
\begin{equation}
	\di\Z = \left(-\frac{1}{2}\gamma^2\Lambda^\dagger\Lambda\, \di t+ \Lambda\, \di X_t\right)\bigl(\Z-c\bigr),
\end{equation}
where $\{\Lambda, P\}$ is the eigensystem diagnonalizing both $A$ and $B$, $\Z=P^{-1}\V$, and $c$ is a vector of constants depending on $a$ and $b$. The solution of this stochastic differential equation can, by Ito's formula \cite{ito}, be seen to be
\begin{equation}
    \Z=\exp(\Lambda (X_t-X_0))(\Z_0-c)+c,
\end{equation}
consequently giving,
\[
	\V = P(\exp(\Lambda (X_t-X_0))(z_0-c)+c),
\]
from which we obtain an expression for the fidelity $\F$.

\subsection{Pauli noise}

We consider the case for Pauli noise. Using stochastic calculus, we can derive the ODE system \eqref{eq:odesystem} for the vector
\[
	\V=\begin{bmatrix}
	\F \\ |\phi^\dagger S\psi|^2 \\ i(\phi^\dagger S \psi\psi^\dagger \phi-\phi^\dagger S \psi\psi^\dagger \phi)
\end{bmatrix}\in\mathbb{R}^3,
\]
obtaining $a=b=0$,
\begin{equation}
    A=\gamma^2\begin{bmatrix}
-1 & 1 & 0  \\
1 & -1 & 0 \\
0 & 0 & -2
\end{bmatrix},\;\; B=\begin{bmatrix}
0 & 0 & -1  \\
0 & 0 & 1 \\
2 & -2 & 0
\end{bmatrix}.
\end{equation}
Notice that $A =\gamma^2 B^2/2$. After applying the diagonalization method, we deduce the following expression for the fidelity
\begin{equation}
\label{eq:paulidist}
\begin{aligned}
    \F &=\frac{1}{2}(1+\S_0^2)+\frac{1}{2}(1-\S_0^2)\cos(2(X-X_0)) \\
    &=\cos^2(X-X_0) + \S_0^2\sin^2(X-X_0),
\end{aligned}
\end{equation}
where $\S_0=\phi_0^\dagger S \phi_0$. For the (WN) and (OU) noise process, we deduce explicit formulas for the expectation and variance for the fidelity (cf.\ App.~\ref{app:cosexp}).

\smallskip

\paragraph*{Ornstein-Uhlenbeck (OU) process.}

For calibrated initial data (i.e.\ $X_0=0$), we find that 
\begin{equation}
    \begin{aligned}
    \label{eq:pauliexp}
        \mathbb{E}[\F_t^{\text{OU}}]&=\frac{1}{2} \left(1+\S_0^2\right)+\frac{1}{2} \left(1-\S_0^2\right) e^{-2\gamma^2\tau_k(t)},\\
        \text{Var}(\F_t^{\text{OU}})&=\frac{1}{8} \left(1-\S_0^2\right)^2 \left[1-e^{-4\gamma^2\tau_k(t)}\right]^2,
    \end{aligned}
\end{equation}
with $\tau_k(t) := e^{-k t} \sinh (k t)/{k}$,
while for stationary initial data (i.e.\ $X_0\sim \gamma\mathcal{N}/\sqrt{2k}$), we obtain
\begin{equation}
    \begin{aligned}
        \mathbb{E}[\F_t^{\text{OU}}]&=\frac{1}{2} \left(1+\S_0^2\right)+\frac{1}{2} \left(1-\S_0^2\right) e^{-2
   (1-e^{-kt})\frac{\gamma^2}{k}},\\
        \text{Var}(\F_t^{\text{OU}})&=\frac{1}{8} \left(1-\S_0^2\right)^2 \left[1-e^{-4(1-e^{-kt})\frac{ \gamma^2}{k}}\right]^2.
    \end{aligned}
\end{equation}

\paragraph*{White noise (WN) process.}
Passing $k\downarrow 0$ in \eqref{eq:pauliexp} and using the fact that $\lim_{k\downarrow 0}\tau_k(t)=t$, we recover the expectation and variance for the (WN) case:
\begin{equation}
    \begin{aligned}
        \mathbb{E}[\F_t^{\text{WN}}]&=\frac{1}{2} \left(1+\S_0^2\right)+\frac{1}{2} \left(1-\S_0^2\right) e^{-2\gamma^2 t},\\
        \text{Var}(\F_t^{\text{WN}})&=\frac{1}{8} \left(1-\S_0^2\right)^2 \left[1-e^{-4\gamma^2t}\right]^2.
    \end{aligned}
\end{equation}

Notice the difference between the large-time asymptotic fidelity values for the (WN) and (OU) processes: 
\begin{align*}
	\lim_{t\to\infty}\mathbb{E}[\F_t^{\text{OU}}] &= \frac{1}{2} \left(1+\S_0^2\right)+\frac{1}{2} \left(1-\S_0^2\right) e^{-\frac{\gamma^2}{k}} \\
	&\ge \frac{1}{2} \left(1+\S_0^2\right) = \lim_{t\to\infty}\mathbb{E}[\F_t^{\text{WN}}],
\end{align*}
and 
\[
	\lim_{t\to\infty}\text{Var}[\F_t^{\text{OU}}] \le \lim_{t\to\infty}\text{Var}[\F_t^{\text{WN}}],
\]
indicating a better limit fidelity for the (OU) process, which is logical given that the (OU) noise is inherently damped.

\subsection{Projection noise}
Next, we consider projection noise. For the vector
\[
	\V=\begin{bmatrix}
		\F \\ \phi^\dagger S\psi\psi^\dagger\phi+\phi^\dagger\psi\psi^\dagger S\phi \\ i(\phi^\dagger S \psi\psi^\dagger \phi-\phi^\dagger S \psi\psi^\dagger \phi)
	\end{bmatrix}\in\mathbb{R}^3,
\]
we find the ODE system \eqref{eq:odesystem} with
\begin{equation}
\begin{aligned}
    A=-\frac{\gamma^2}{2}\begin{bmatrix}
0 & 1 & 0  \\
0 & 1 & 0 \\
0 & 0 & 1
\end{bmatrix}&,\;\; B=\begin{bmatrix}
0 & 0 & 1  \\
0 & 0 & 1 \\
0 & 1 & 0
\end{bmatrix},\\
a=\gamma^2\S_0^2\begin{bmatrix}
1 \\
1 \\
0
\end{bmatrix}&,\;\;
b=\S_0^2\begin{bmatrix}
0  \\
0 \\
-2
\end{bmatrix}.
\end{aligned}
\end{equation}

The diagonalization method then gives the fidelity
\begin{equation}
    \F = 1 - 2 (1 - \S_0^2) \S_0^2 (1 - \cos(X_t - X_0)).
\end{equation}
As in the previous case, we can deduce explicit formulas for the expectation and variance for the fidelity under (WN) and (OU) noise processes (cf.\ App.~\ref{app:cosexp}).

\smallskip
\paragraph*{Ornstein-Uhlenbeck (OU) process.}
For calibrated initial data (i.e.\ $X_0=0$), we find
\begin{equation}
\begin{aligned}\label{eq:projexp-calibrated}
    \mathbb{E}[\F_t^{\text{OU}}] &= 1-2 (1-\S_0^2) \S_0^2 \left[1-e^{-\frac{\gamma^2}{2}\tau_k(t)}\right],\\
   \text{Var}(\F_t^{\text{OU}}) &= 2 (1-\S_0^2)^2 \S_0^2 \left[1-e^{-\gamma^2\tau_k(t)}\right]^2,
\end{aligned}
\end{equation}
while for stationary initial data (i.e.\ $X_0\sim \gamma\mathcal{N}/\sqrt{2k}$), we find
\begin{equation}
\label{eq:projexp}
\begin{aligned}
    \mathbb{E}[\F_t^{\text{OU}}] &= 1-2 (1-\S_0^2) \S_0^2 \left[1-e^{-(1-e^{-k t})\frac{\gamma^2}{2 k}}\right],\\
   \text{Var}(\F_t^{\text{OU}}) &= 2 (1-\S_0^2)^2 \S_0^2 \left[1-e^{-(1-e^{-k t})\frac{\gamma^2}{k}}\right]^2.
\end{aligned}
\end{equation}
\paragraph*{White noise (WN) process.}
Passing $k\downarrow 0$ in \eqref{eq:projexp-calibrated}, one recovers the expectation and variance for the (WN) case:
\begin{equation}
    \begin{aligned}
        \mathbb{E}[\F_t^{\text{WN}}]&=1-2 (1-\S_0^2) \S_0^2 \left[1-e^{-\frac{\gamma^2}{2}t}\right],\\
        \text{Var}(\F_t^{\text{WN}})&=2 (1-\S_0^2)^2 \S_0^2 \left[1-e^{-\gamma^2t}\right]^2.
    \end{aligned}
\end{equation}
Similar to Pauli noise, the large time-asymptotics show that the limit expected fidelity value under (OU) noise is higher than that of the (WN) noise.

\subsection{Non-commuting}
\label{sec:noncommute}
For the non-commuting case, we could still find a closed ODE system of size $m=10$ due to the group structure of the Pauli matrices. We find a system of 10 equations where the matrix $A$ splits as $A=A_c + B^2$, such that
\begin{equation}
    \di\V = \alpha A_c \V\, \di t + \frac{\gamma^2}{2} B^2 \V\,\di t + B \V\, \di X_t,
\end{equation}
with $A_c$ the commutator terms and $B^2$ the noise terms (see App.~\ref{app:noncommute} for full matrices). In the absence of $A_c$, the system can be solved using the diagonalization method outlined in Sec.~\ref{sec:odes}. However, because $A_c$ and $B$ do not commute, this is not possible in the full system. Instead, assuming that the driving Hamiltonian dominates the noise intensity, i.e.\ $\gamma^2/\alpha \ll 1$, we can use the stochastic Magnus expansion \cite{stochasticmagnus} (see App.~\ref{app:noncommute}) to obtain an approximation fo the fidelity.

\smallskip
\paragraph*{White noise (WN) process.} For the white noise case,
\begin{equation}
\label{eq:noncommuteexp}
    \begin{aligned}
    \mathbb{E}[\F]\approx&=\frac{1}{2}+\frac{1}{2} \C_1^2 e^{-2 \gamma^2 t}\\
    &+\frac{1}{2} e^{-\gamma^2 t} \Bigl(\sinh (u) ((\C_2^2-\C_3^2) \cos (2a t)\\
    &-2 \C_2 \C_3
   \sin (2a t))+\left(\C_2^2+\C_3^2\right) \cosh (u)\Bigr)
\end{aligned}
\end{equation}
where $\C_i=\phi_0^\dagger\sigma_i\phi_0$ with $\C_1^2+\C_2^2+\C_3^2=1$. From this expression, one can observe the oscillatory behavior introduced by the Hamiltonian. This can be explained by considering $H=X, S=Y$ and $\psi_0=|0\rangle$. Initially, the state is susceptible to the $S=Y$ noise and will start rotating. As the evolution under $H$ progresses, the state moves to $(|0\rangle+|1\rangle)/\sqrt{2}$, which is not susceptible to this noise and even later moves to $|1\rangle$, which is again susceptible. This process repeats, causing the oscillatory behavior with frequency dependent on the strength of the Hamiltonian $\alpha$ (see Fig.~\ref{fig:noisyblochsphere}). 

\smallskip
\paragraph*{Ornstein-Uhlenbeck (OU) process.} Unfortunately, we were unable to obtain a satisfying approximation for the expectation of the fidelity with (OU) noise process. Nevertheless, we provide in App.\ref{app:noncommute} a very rudimentary approximation that may obtain nonphysical values for certain parameters.

\subsection{Multi-qubit}
In a work by Kobayashi and Yamamoto~\cite{boundspaper}, noise acting on multiple qubits is analyzed for multi-qubit operators satisfying $S^2=I$. These are important cases, as noise on the $S=Z\otimes Z$ operator could model noise during an entanglement procedure, while $S=X\otimes I$ or $S=I\otimes X$ represents noise acting on one of the qubits in the system. However, it is specifically mentioned in \cite{boundspaper} that operators of the form $I\otimes Q+Q\otimes I$, with $Q^2=I$, cannot be treated using their methods. These types of operators are also important cases, as they model noise on a global control addressing all qubits simultaneously, which is prevalent in many quantum computing systems \cite{babalar}.

\smallskip

With our methods, we can solve these systems by defining $R := Q\otimes Q$, such that
\begin{align}
S^2 =I\otimes Q^2+Q^2\otimes I+2R.
\end{align}
For Pauli noise ($Q^2=I$) we find
\begin{equation}
\begin{aligned}
\label{eq:multiqubit}
    \F &= \frac{1}{4} \left(\S_0^2+(\R_0-1)^2+2\left(1-\R_0^2\right) \cos (2 (X_t-X_0))\right.\\
    &+\left.(1-\S_0+\R_0) (1+\S_0+\R_0) \cos ^2(2 (X_t-X_0))\right),
\end{aligned}
\end{equation}
with $\S_0=\phi_0^\dagger S\phi_0$ and $\R_0=\phi_0^\dagger R\phi_0$. 

\smallskip
For projection noise ($Q^2=Q$), we find 
\begin{equation}
    \begin{aligned}
        \F &=1-2 \big(\cos (X_t-X_0)-1\big) \times\\
        &\quad\big(\S_0^2+\S_0 (2 \R_0 (\cos (X_t-X_0)-1)-1)\\
        &-2 \R_0 [\R_0
   (\cos (X_t-X_0)-1)+\cos (X_t-X_0)]\big).
    \end{aligned}
\end{equation}
If we assume that the two-qubit system starts in a product state $\phi_{0}=\phi_{0}^{(0)}\otimes\phi_{0}^{(1)}$, then the values $\S_0$ and $\R_0$ can be expressed as $\S_0=\S_{0}^{(0)}+\S_{0}^{(1)}$ and $\R_0=\S_0^{(0)}\S_0^{(1)}$, where the superscript $(i)$ refers to the $i$-th qubit.  For both Pauli and projection noise, the fidelity distribution then factors as
\begin{equation}
\label{eq:factor}
    \F=\F^{(0)}\F^{(1)}.
\end{equation}
In App.~\ref{app:multiq}, this is generalized to $n$-qubit systems. Physically, this result is intuitive, as the two qubits evolve independently but under the same noise. This means that their fidelities are only dependent on the noise process and thus factor on the level of the distribution. Note that this does not mean that the expectation of the fidelity will factor as $\F^{(0)}$ and $\F^{(1)}$ are not independent. The full expectation $\mathbb{E}[\F]$ can still be calculated according to the methods of App.~\ref{app:cosexp}. Moreover, the two-qubit method can also treat entangled ground states for which the distributions do not factor as described by Eq.~\eqref{eq:factor}.

\section{Simulation results}
\label{sec:results}
In this section, we verify our fidelity calculations by comparing them with Monte Carlo simulations. For all results in this work, we employ the explicit second-order scheme due to Platen \cite{stochasticintegration2,platen} as our Monte-Carlo stochastic integration method, as detailed in App.~\ref{app:stochint}. We also take the driving Hamiltonian $H=0$, unless explicitly stated otherwise.

\subsection{Pauli noise + Order of Approximation}
\label{sec:resultsa}
We start by showing that the order of approximation method, as detailed in Sec.~\ref{sec:odes}, becomes better with increasing order, but eventually leads to nonphysical behavior. For the Pauli noise system, we get an evolution for the vector $\x=\mathbb{E}[\V]$ given by
\begin{equation}
    \dot \x= \begin{bmatrix}
-\gamma^2 & \gamma^2 & ik  \\
\gamma^2 & -\gamma^2 & -ik \\
2ip & -2ip & -k-2\gamma^2
\end{bmatrix} \x,\quad \x(0)=\begin{bmatrix}
1  \\
\S_0^2  \\
0 
\end{bmatrix}
\end{equation}
with $p=(k\mathbb{E}[X_t^2]-\gamma^2)$ and $\S_0=|\phi_0^\dagger S \phi_0|$. Here, we made the assumption $\mathbb{E}[X_t^2\V]=\mathbb{E}[X_t^2]\mathbb{E}[\V]$ to close the system. The second-order system is detailed in App.~\ref{app:secondorder}. Using the fact that $\mathbb{E}[X_t^2]=(1-e^{-2 kt})\gamma^2/2k$ (see App.~\ref{app:cosexp}), we can solve these systems numerically using the \texttt{odeint} methods from SciPy \cite{2020SciPy-NMeth}.

\smallskip

From Fig.~\ref{fig:orderofapprox}, we see that for high fidelities the order of approximation method works well, and that indeed the second-order solution follows the exact solution longer than the first-order does. This shows that the Ito isometry \cite{itoisometry} approximation of App.~\ref{app:itoiso} works satisfactorily for short durations. However, both orders eventually lead to nonphysical behavior, as fidelities are predicted outside the interval $[0,1]$, and thus do not correctly predict the long-time behavior. This will become especially important as one starts looking at optimal control methods for the SSE, where the stability of the fidelity needs to be maintained for long periods. This showcases the importance of our full distribution method, which correctly predicts the long-time behavior of the fidelity.

\begin{figure}[t]
    \centering
    \includegraphics[scale=0.55]{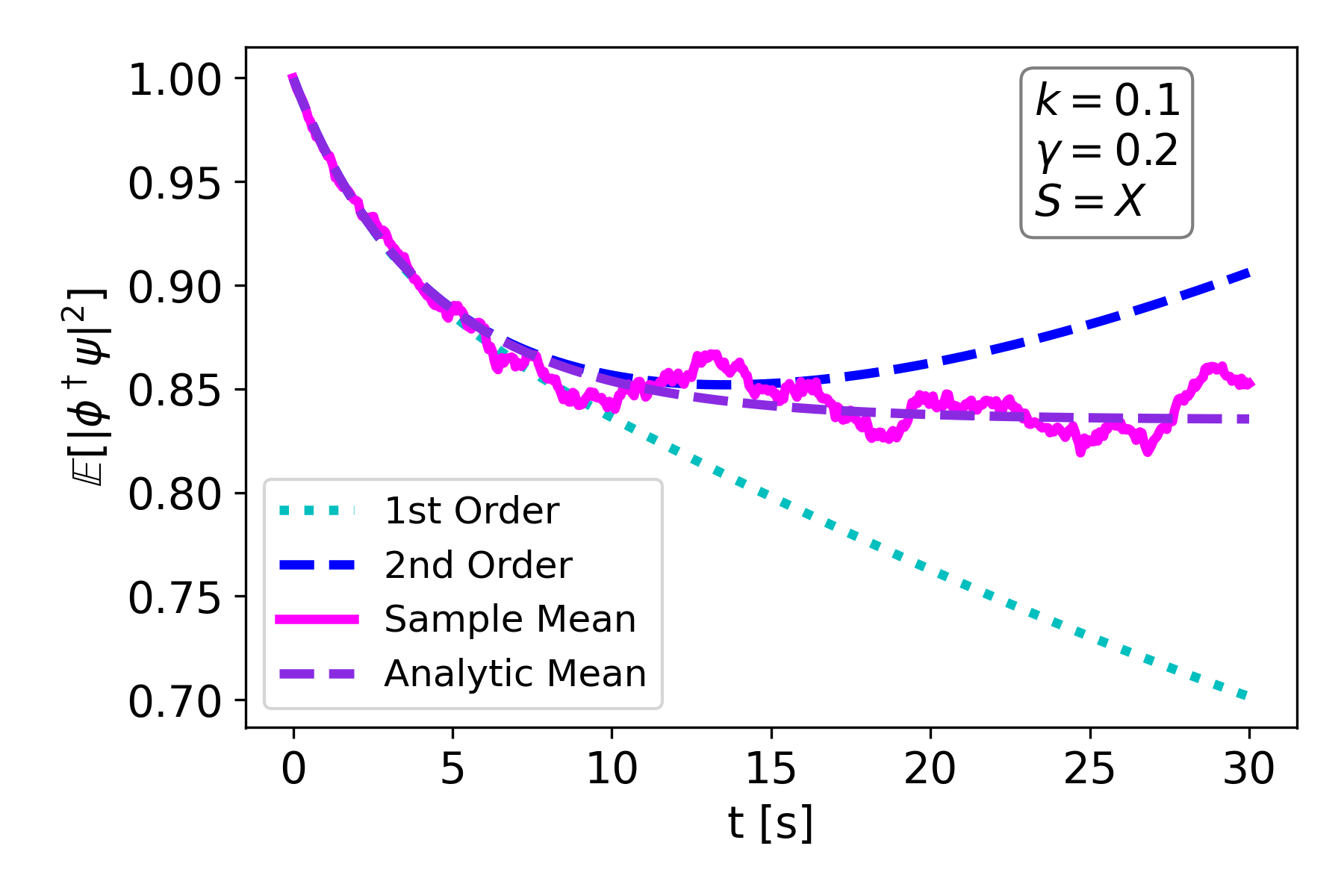}
    \caption{\emph{Pauli noise:} Fidelity expectations for a 1 qubit system with $S=X, \gamma=0.2$, $k=0.1$ and calibrated initial data. First and second order of approximation together with sample mean and analytic mean as in Eq.~\eqref{eq:pauliexp}. Computed with $2\cdot10^2$ Monte-Carlo simulations.}
    \label{fig:orderofapprox}
\end{figure}

\begin{figure}[b]
    \centering
    \includegraphics[scale=0.4]{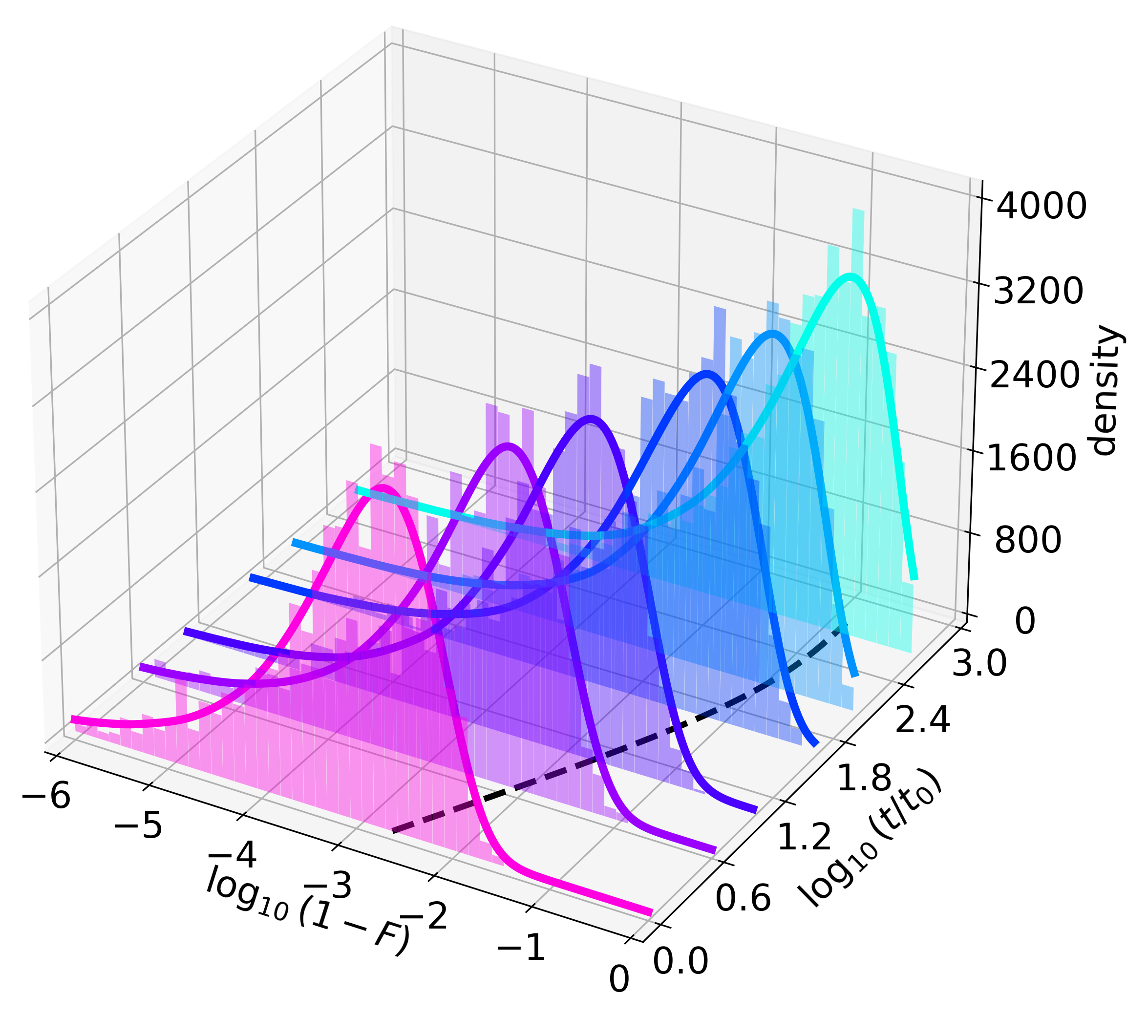}
    \caption{\emph{Pauli noise:} Fidelity distributions at various times $t$ for a 1 qubit system with $S=X, \gamma=0.2$, $k=0.1$ and calibrated initial data as in Eq.~\eqref{eq:paulidist}. Solid lines are kernel density estimates of distributions as in Eq.~\eqref{eq:paulidist} and histograms are from Monte Carlo simulations of Eq.~\eqref{eq:sse-ou}. Here $t_0=0.06$ is an arbitrary timescale. Expectation of fidelity as dashed black line. $2\cdot10^3$ Monte-Carlo simulations for each curve and histogram.}
    \label{fig:distributions}
\end{figure}

\smallskip

In Fig.~\ref{fig:distributions}, we see the full distributions of the fidelity as a function of time. Shown here, are histograms for the fidelities which are the result of $2\cdot10^3$ Monte-Carlo SSEs (taking about 30 mins on an 8-core laptop). In contrast, we also determine the kernel density estimate \cite{kde} of $1\cdot 10^3$ (OU) noise realizations (taking mere seconds on the same setup), resulting in the solid distribution lines. This shows that our distribution method gives the correct result, while also being significantly less computationally demanding. Furthermore, from Fig.~\ref{fig:distributions} we see that the distributions maintain their shape, while their peaks shift very quickly towards zero fidelity.

\subsection{Projection noise}

For the projection noise, we take $S=|1\rangle\langle 1|$ and analyze the fidelity for both the (OU) and (WN) noise. As Fig.~\ref{fig:projnoise} shows, (OU) noise and (WN) noise initially have similar fidelity but quickly diverge, leading to differing asymptotic values. This is to be expected, as the second moment of (WN) grows as $\mathbb{E}[X_t^2]=\gamma^2\sqrt{t}$, whereas for (OU) noise $\lim_{t\rightarrow\infty}\mathbb{E}[X_t^2]=\gamma^2/2k$. Thus, the (WN) process drifts off, whereas the (OU) noise keeps fluctuating around the zero value, constantly correcting for itself. We also see this in Fig.~\ref{fig:projnoise}, where the variance of the fidelity evolving under (WN) noise grows much faster than that of evolving under(OU) noise.

\begin{figure}[t]
    \centering
    \includegraphics[scale=0.55]{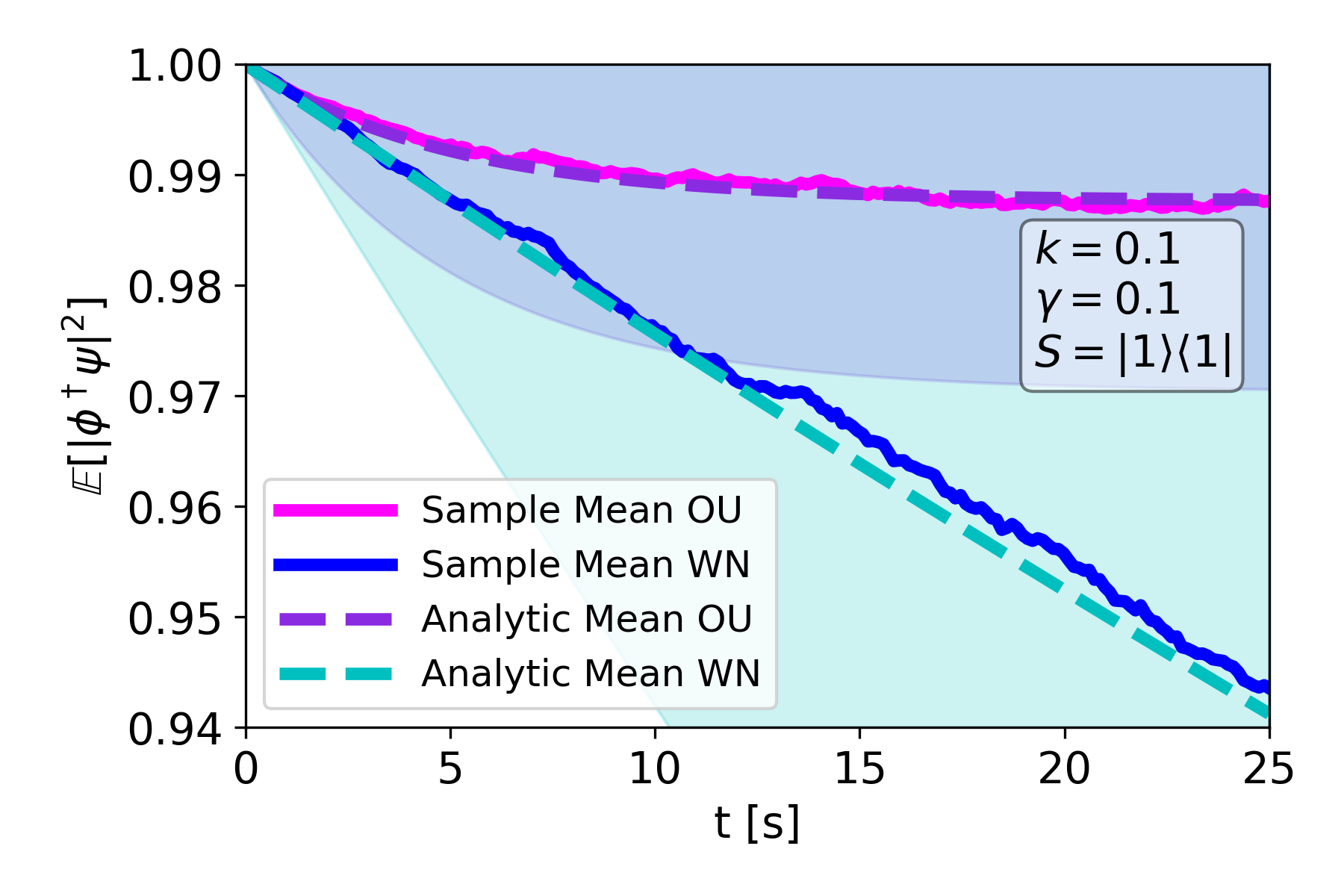}
    \caption{\emph{Projection noise: }Fidelity expectations and variances for a 1 qubit system with $S=|1\rangle\langle1|$, $\gamma=0.1$, $k=0.1$ (and $k=0$ for (WN)) and stationary initial data, with sample mean and analytic mean as in Eq.~\eqref{eq:projexp} and analytic variances. Computed with $5\cdot10^2$ Monte-Carlo simulations.}
    \label{fig:projnoise}
\end{figure}

\subsection{Non-commuting}
We analyze the non-commuting case as described in Sec.~\ref{sec:noncommute}. From Fig.~\ref{fig:noncommute}, we see that the sample mean fidelities nicely overlap with the predicted result from Eq.~\eqref{eq:noncommuteexp}. The fidelities indeed show the oscillating behavior caused by the non-commuting Hamiltonian, which consistently cycles through the susceptibility of the state to the noise. The long-time behavior matches the Monte-Carlo simulations because of the stochastic Magnus expansion type of approximation.

\smallskip
Furthermore, we see that for $\C_1=1$ the initial state is an eigenstate of the Hamiltonian, meaning that it will not cycle in and out of susceptibility to the noise and thus decrease in fidelity exponentially with a rate $-2\gamma^2$. On the other hand, for $\C_2=1$ or $\C_3=1$ the state, under the influence of the Hamiltonian evolution, cycles in and out of susceptibility to the noise and only decreases exponentially with a rate $-\gamma^2$, i.e.\ less than the previous scenario. Also note from the inset of Fig.~\ref{fig:noncommute}, the difference between $\C_2=1$ and $\C_3=1$, where for $\C_2=1$ the initial state is an eigenstate of the noise operator, meaning that it is not susceptible to noise and thus the fidelity does not immediately decrease. For $\C_3=1$,  the initial state is the most susceptible to noise and the fidelity immediately drops, following the $\C_1=1$ curve until the Hamiltonian rotation starts to take effect. 

\begin{figure}[t]
    \centering
    \includegraphics[scale=0.5]{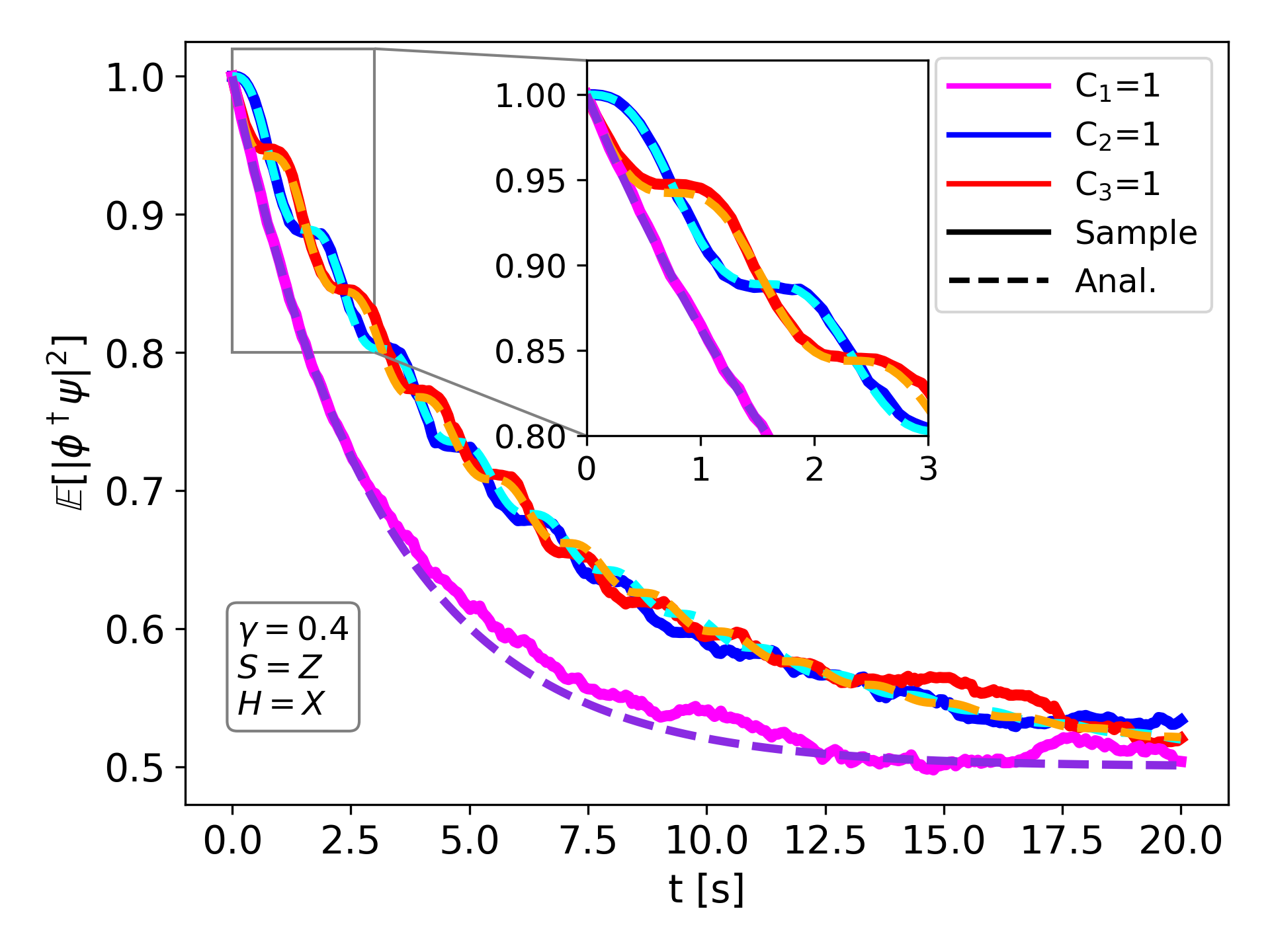}
    \caption{\emph{Non-commuting:} Fidelity expectations for a 1 qubit system with $S=Z$, $H=X$, $\gamma=0.4$ and calibrated initial data. Sample mean and analytic mean as in Eq.~\eqref{eq:noncommuteexp} with varying conditions for $\C_1^2+\C_2^2+\C_3^2=1$. Computed with $1\cdot10^3$ Monte-Carlo simulations.}
    \label{fig:noncommute}
\end{figure}

\subsection{Multi-qubit}
Lastly, we analyze a two-qubit system with the noise operator $S=X\otimes I+I\otimes X$. We analyze both the product initial state $\phi_0=|00\rangle$ (which could have also been treated as the product of the distributions of the individual qubits, c.f. App.~\ref{app:multiq}), as well as the maximally entangled GHZ \cite{ghz} initial state $\phi_0=(|00\rangle+|11\rangle)/\sqrt{2}$. From Fig.~\ref{fig:twoqubit}, we see that our distribution method matches the Monte-Carlo simulation results.

\smallskip

Furthermore, we see that the product initial state indeed yields different results from the entangled GHZ state. Even more interesting is the fact that for large values of $k\gg 1$, the $|00\rangle$ initial state has a better asymptotic fidelity, whereas, for small values of $k\ll 1$, the GHZ initial state is better. Intuitively, this can be explained by considering a unitary operator $U(X_T)$ mapping $|00\rangle\rightarrow |\eta_T\rangle$ and $|11\rangle\rightarrow |\zeta_T\rangle.$ The overlap of the final states with $\phi_T$, which coincides with $\phi_0$ since $H=0$, are given for the case $\phi_0=|00\rangle$ by
\[
    \langle \psi_T|\phi_T\rangle = \langle 00|\eta_T\rangle,
\]
and for the case $\phi_0=\text{GHZ}$ by
\[
    \langle \psi_T|\phi_T\rangle = \frac{1}{2}\bigl(\langle 00|\psi_T\rangle+\langle 11|\zeta_T\rangle+\langle 00|\zeta_T\rangle+\langle 11|\psi_T\rangle\bigr).
\]

For both small and large noises we can, by symmetry, consider $\langle 00|\psi_T\rangle=\langle 11|\zeta_T\rangle$. If $k\gg 1$ is large, the noise will remain small resulting in $|\psi_T\rangle\approx|00\rangle$ and $|\zeta_T\rangle\approx|11\rangle$, and the cross terms barely contribute to the fidelity. Conversely, for small $k\ll 1$, the noise is larger, and the cross terms become large, resulting in higher fidelities. Note that this does not always have to be the case. Indeed, if $S=Y\otimes I+I\otimes Y$, then GHZ is an eigenstate with eigenvalue 0 and does not evolve under this noise at all. 

\begin{figure}[b]
    \centering
    \includegraphics[scale=0.27]{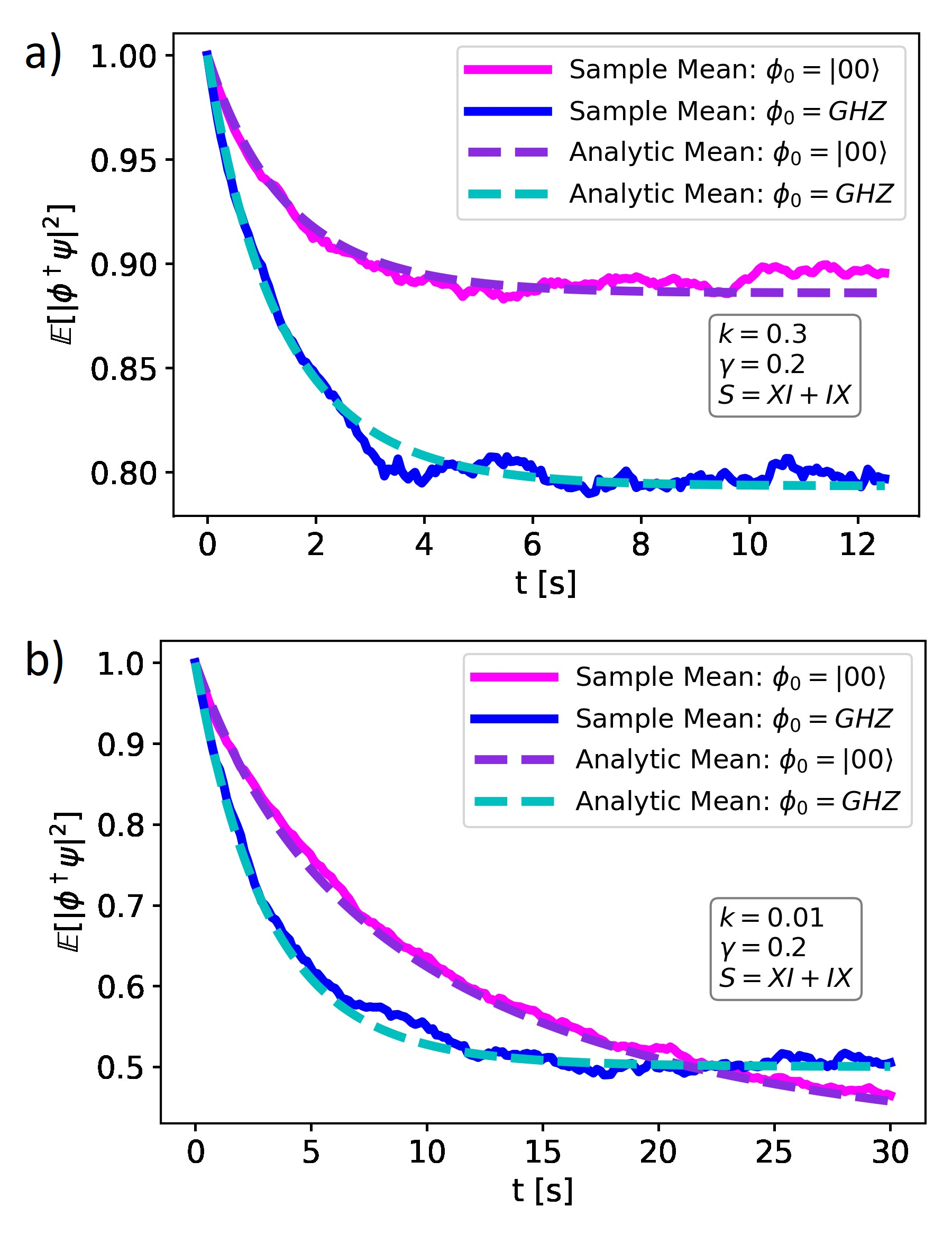}
    \caption{\emph{Two-qubit:} Fidelity expectations for a 2 qubit system with $S=X\otimes I+I\otimes X$, $\gamma=0.2$ and stationary initial data. \emph{(a)} $k=0.3$ and \emph{(b)} $k=0.01$. The sample and analytic means for $|00\rangle$ and GHZ initial state, as calculated from Eq.~\eqref{eq:multiqubit}. Computed with 500 Monte-Carlo simulations.}
    \label{fig:twoqubit}
\end{figure}

\section{Conclusion}
\label{sec:conclusion}
This work describes a novel method for expressing the full distribution of qubit fidelity under SSEs. These solutions contain more information than the standard Lindblad equation, which can in turn be used for noise characterization and optimal control problems. Furthermore, our work is the first to predict the long-term behavior of fidelities evolving under realistic noise types, such as (OU) noise stemming from fluctuation-dissipation processes. This long-time behavior is especially important in control processes, where one seeks to maintain a state faithfully for a long period. The methods developed here are significantly less computationally demanding and more accurate than standard Monte Carlo methods for SSEs.  All these predictions can further be used to decide on the allowed levels of noise in a qubit control system, e.g.\ a driving laser or a cavity resonator, thus, resulting in an accurate benchmark for control system noise on a qubit system.

\smallskip

In future work, we hope to extend this work to predictions on noises with arbitrary power spectral densities. Furthermore, we hope to generalize our method to $N$ qubits and extend the results to non-commuting Hamiltonians. Lastly, we hypothesize that these methods on SSEs can be used to introduce optimal control for noise mitigation on open quantum systems, which so far has failed to be realized for Lindblad equations. 

\section*{Acknowledgements}
We thank Jasper Postema, Raul F. Santos, Madhav Mohan, and Jasper van de Kraats for fruitful discussions. This research is financially supported by the Dutch Ministry of Economic Affairs and Climate Policy (EZK), as part of the Quantum Delta NL program,  the project “EuRyQa - European infrastructure for Rydberg Quantum Computing” grant agreement number 10107014, and by the Netherlands Organisation for Scientific Research (NWO) under Grant No.\ 680.92.18.05. L.Y.\ Visser and O.\ Tse acknowledge support from NWO grant NGF.1582.22.009.

\section*{Data Availability}
The data that support the findings of this study are available from the corresponding author upon reasonable request.

\newpage

\bibliographystyle{apsrev4-1}
\bibliography{Bibliography.bib}

\onecolumngrid

\appendix
\section{Formal derivation of the Stochastic Schr\"{o}dinger equation}
\label{app:ssederivation}
In this section, we follow the derivation of the SSE as in \cite{colorednoisepaper}. Consider an (OU) process $X(t)$ and operators $C,D$, and $R$ on the Hilbert space. The basic linear Schr\"{o}dinger equation with noise is given by
\begin{equation}
    \di \psi=(C+DX_t) \psi\, \di t+R \psi\, \di X_t.
\end{equation}
Using $\di X_t=-kX_t+\gamma \di W_t$, this can be rewritten as
\begin{equation}
    \di \psi=(C+X D-k X R) \psi\, \di t+ \gamma R \psi\, \di W_t.
\end{equation}
For normalization, $\di |\psi|^2=(\di \psi^\dagger)\psi+\psi^\dagger(\di\psi)+(\di\psi^\dagger)(\di\psi)=0$ is required at all times. Expanding this using Ito calculus ($\di t^2=0, \di W_t\di t=0, \di W_t^2=\di t)$ \cite{ito} gives
\begin{equation}
\di|\psi|^2=\psi^\dagger\left[C^{\dagger}+C+X_t\left(D^{\dagger}+D-k R-k R^{\dagger}\right)+\gamma^2 R^{\dagger} R\right]\psi \,\di t+\gamma \psi^\dagger (R^{\dagger}+R) \psi \,\di W_t=0.
\end{equation}
Thus, $C^\dagger+C+R^\dagger R=0$, $D^\dagger+D-kR-kR^\dagger=0$ and $R^\dagger+R=0$. One can choose $R=iS$ with $S=S^\dagger$ Hermitian. Furthermore, $C=-iH-\frac{1}{2}S^\dagger S=-iH-\frac{1}{2}S^2$, with $H=H^\dagger$ and $D=0$ to finally get the SSE for an (OU) process as
\begin{equation}
    \text{(OU)}:\quad \di \psi=-iH\psi \,\di t+i k X_t S\psi \,\di t-\frac{\gamma^2}{2} S^\dagger S \psi \,\di t - i\gamma S\psi \,\di W_t.\\
\end{equation}
Letting $k\rightarrow0$ results in the SSE for the (WN) process as 
\begin{equation}
    \text{(WN)}:\quad \di \psi=-iH\psi \,\di t-\frac{\gamma^2}{2} S^\dagger S \psi \,\di t - i\gamma S\psi \,\di W_t.\\
\end{equation}
Note that Ito calculus holds for any semimartingale \cite{ito}, and thus for general semimartingale noise we can write 
\begin{equation}
    \text{(SM)}:\quad \di \psi=-iH\psi \,\di t-\frac{1}{2}S^\dagger S \psi \,\di[X]_t - i S\psi \,\di X_t,\\
\end{equation}
where $[X]_t$ is the quadratic variation of the process \cite{quadvar}. By defining $\rho=\mathbb{E}[\psi\psi^\dagger]$ and noting that $\mathbb{E}[\di W_t]=0$
\begin{equation}
\begin{aligned}
\di \psi\psi^\dagger&=-i[H,\psi\psi^\dagger]\,\di t+(S\psi\psi^\dagger S^\dagger-\frac{1}{2}\{S^\dagger S,\psi \psi^\dagger\})\,\di [X]_t-i[S,\psi\psi^\dagger]\,\di X_t\\
\Rightarrow \partial_t\rho &=-i[H,\rho]+\gamma^2(S\rho^\dagger S^\dagger-\frac{1}{2}\{S^\dagger S,\rho\})-i[S,\mathbb{E}[\psi\psi^\dagger \,\di X_t]],
\end{aligned}
\end{equation}
which for white noise reduces to the Lindblad equation since $\mathbb{E}[\psi\psi^\dagger \di X_t]=\mathbb{E}[\psi\psi^\dagger \di W_t]=\mathbb{E}[\psi\psi^\dagger]\mathbb{E}[ \di W_t]=0$. Note that for general noise, there is no independence of the state and the noise increments (i.e. $\mathbb{E}[\psi\psi^\dagger \di X_t]\neq\mathbb{E}[\di X_t]\rho$), and we do not get the standard Lindblad equation.

\section{Ito's Isometry}
\label{app:itoiso}
We detail the approximation method for terms of the form $\mathbb{E}[X_t^2 \V]$ as described in Sec.~\ref{sec:odes}. The integral solution of the (OU) process \cite{ounoise} can be written as
\begin{equation}
    X_t=\int_0^t \gamma e^{k(s-t)} \di W_s.
\end{equation}
By positivity of $\V$
\begin{equation}
    \mathbb{E}[X_t^2 \V]=\mathbb{E}\left[\left(\int_0^t \gamma e^{k(s-t)} \di W_s\right)^2\V\right]=\mathbb{E}\left[\left(\int_0^t \gamma \sqrt{\V} e^{k(s-t)} \di W_s\right)^2\right].
\end{equation}
If $\V$ were to be adapted to the natural filtration of the Wiener process $W_t$, Ito's isometry \cite{itoisometry} can be used to write
\begin{equation}
    \mathbb{E}[X_t^2 \V]=\mathbb{E}\left[\int_0^t \gamma^2 \V e^{2k(s-t)} \di t\right]=\frac{\gamma^2}{2k}(1-e^{-2kt})\mathbb{E}[\V]=\mathbb{E}[X_t^2]\mathbb{E}[\V].
\end{equation}
However, $\V$ could for instance be the fidelity at time $t$ which is dependent on the Wiener process in the interval $[0,t]$. Therefore, Ito's isometry does not hold. Nevertheless, for small values of $t$, the $\V$ values are roughly equal to their respective initial values, and Ito's isometry can be used as an approximation.  

\section{Moment calculation}
\label{app:cosexp}
For the expectation and variance of the fidelity distributions, expressions for the expectations of terms $\cos(\alpha(X_t-X_0))$ for $\alpha>0$ have to be calculated (note that for square terms $\cos(x)^2=(\cos(2x)+1)/2$). To do so, the cosines are expanded into their power series representation, and both linearity and conditional expectations are used to get
\begin{equation}
\label{eq:expcos}
    \mathbb{E}[\cos(\alpha(X_t-X_0))]=\sum_{n=0}^\infty \frac{(-1)^n}{2n!}\alpha^{2n}\mathbb{E}\big[\mathbb{E}[(X_t-X_0)^{2n}|X_0]\big].
\end{equation}

The (OU) process, by Ito's formula \cite{itoisometry}, can be shown to obey
\begin{equation}
\label{eq:difou}
    \di X_t^{n}=-nkX^{n}\di t+n\gamma X^{n-1}\di W_t+\frac{1}{2}(n^2-n)\gamma^2X^{n-2}\di t, \quad X^n(0)\sim X_0^n,\quad \forall n\in \mathbb{N}.
\end{equation}

\subsection*{Calibrated initial data $X_0=0$}

For calibrated initial data (e.g. $X_0=0$), it is easily proven by induction that
\begin{equation}
\label{eq:expcalibrated}
    \mathbb{E}\big[\mathbb{E}[(X_t-X_0)^{2n}|X_0]\big]=\mathbb{E}[X_t^{2n}]=\frac{\Gamma \left(n+\frac{1}{2}\right) }{\sqrt{\pi }}\left(\frac{2\gamma^{2}}{k}e^{-k  t} \sinh (k t)\right)^n.
\end{equation}
where $\Gamma$ is the gamma function. Resulting from Eq.~\eqref{eq:expcos} in
\begin{equation}
    \mathbb{E}[\cos(\alpha(X_t-X_0))]=\exp\left(-\alpha^2\frac{\gamma^2}{2k} e^{-k t} \sinh(kt) \right).
\end{equation}

\subsection*{Stationary initial data $X_0\sim \gamma\mathcal{N}/\sqrt{2k}$}

For stationary initial data, a slightly more involved approach is necessary, where the binomial theorem is employed to get
\begin{equation}
\label{eq:binomial}
    \mathbb{E}[(X_t-X_0)^{2n}|X_0]=\sum_{m=0}^{2n}\binom{2n}{m}(-1)^{2n-m}X_0^{2n-m}\mathbb{E}[X_t^m|X_0].
\end{equation}
From the differential equations in Eq.~\eqref{eq:difou}
\begin{equation}
    d\mathbb{E}[X_t^m|X_0]=-mk\mathbb{E}[X_t^m|X_0]\di t+\frac{1}{2}(m^2-m)\gamma^2\mathbb{E}[X_t^{m-2}|X_0]\di t,
\end{equation}
with initial conditions $\mathbb{E}[X_t^0|X_0]=1,\, \mathbb{E}[X_t^1|X_0]=X_0e^{-kt}.$ This can be shown to have the solution
\begin{equation}
\label{eq:expxtm}
\begin{aligned}
\mathbb{E}[X_t^m|X_0]&=\begin{cases}e^{-2wkt}(2k)^{-w}
 \sum_{l=0}^{w}a[l, w](e^{2kt}-1)^{w - l}\gamma^{2(w - l)}k^l
   X_0^{2l},\quad\,\,\,\,\, w=m/2,\qquad \,\,\,\text{ for } m \text{ even}
   \\e^{-(2w+1)kt}(2k)^{-w}
   \sum_{l=0}^{w}b[l, w](e^{2kt}-1)^{w - l}\gamma^{2 (w - l)}k^l
     X_0^{2 l+1}, \quad w=(m-1)/2, \text{ for } m \text{ odd}
\end{cases},\\
&a[l,w] := 2^l\prod_{q=l+1}^w\frac{q (2 q - 1)}{q - l}
   ,\quad b[l,w] := 2^l\prod_{q=l}^{w-1}\frac{(1 + q) (3 + 2 q)}{(q - l + 1)}.
\end{aligned}
\end{equation}

For verification, the expressions for $\mathbb{E}[X_t^{2m}]$ in the case of calibrated initial data (Eq.~\eqref{eq:expcalibrated}) are retrieved when taking $X_0=0$. Using conditional expectation on Eq.~\eqref{eq:binomial} gives

\begin{equation}
\label{eq:expxtmx0}
    \mathbb{E}\big[\mathbb{E}[(X_t-X_0)^{2n}|X_0]\big]=\sum_{m=0}^{2n}\binom{2n}{m}(-1)^{2n-m}\mathbb{E}[X_0^{2n-m}\mathbb{E}[X_t^m|X_0]].
\end{equation}
Furthermore, the normal distribution of $X_0$ gives
\begin{equation}
\label{eq:normaldis}
    \mathbb{E}[X_0^q]=\begin{cases} \left(\frac{\gamma}{\sqrt{2k}}\right)^q(q-1)!!, \quad q \text{ even }
   \\0, \quad q \text{ odd }.
\end{cases}
\end{equation}
Combining Eq.~\eqref{eq:expxtm} and Eq.~\eqref{eq:normaldis}, and filling into Eq.~\eqref{eq:expxtmx0} gives
\begin{equation}
    \mathbb{E}[\mathbb{E}[(X_t-X_0)^{2n}|X_0]]=(2 n - 1)!!\left(\frac{\gamma^2}{k}e^{-kt}(e^{kt}-1)\right)^n,
\end{equation}
which, when used in Eq.~\eqref{eq:expcos}, results in
\begin{equation}
    \mathbb{E}[\cos(\alpha(X_t-X_0))]=\exp\left(-\alpha^2 \frac{\gamma^2}{2k} e^{-k t} (e^{k t}-1)\right).
\end{equation}

\section{Multi-qubit noise}
\label{app:multiq}
This section proves a lemma regarding the factoring of fidelity distributions of $n$-qubit systems, which have a product initial state and evolve under the same noise source.
\begin{lemma}[Factoring fidelity for pure states]
Let the state $\psi_{(n)}$ of a $n$-qubit system evolve according to
\begin{equation}
    \di \psi_{(n)}=-iH_n\psi_{(n)}\di t-\frac{1}{2}S_n^2\psi_{(n)}\di [X]_t-iS_n\psi_{(n)} \di X_t,\quad \psi_{(n)}(0)=\bigotimes_{j=1}^n \psi_{j0},
\end{equation}
where the Hamiltonian $H_n$ and the noise operator $S_n$ take the sum form
\begin{equation}
\begin{aligned}
    H_n&=\sum_{j=1}^n A_j,\quad A_j=I^{\otimes (j-1)}\otimes \tilde{A}_j \otimes I^{\otimes (n-j)},\quad \tilde{A_j}=\tilde{A_j}^\dagger,\\
    S_n&=\sum_{j=1}^n Q_j,\quad Q_j=I^{\otimes (j-1)}\otimes \tilde{Q}_j \otimes I^{\otimes (n-j)},\quad \tilde{Q_j}=\tilde{Q_j}^\dagger.
\end{aligned}
\end{equation}
Let $\phi_{(n)}$ be the noiseless target state. Then 
\[
	|\phi_{(n)}^\dagger \psi_{(n)}|^2=\prod_{j=1}^n |\phi_j^\dagger \psi_j|^2,
\] where $\psi_j$ is a 1-qubit state evolving according to the SSE
\begin{equation}
        \di \psi_j=-i\tilde{A}_j\psi_j\,\di t-\frac{1}{2} \tilde{Q}_j^2\psi_j \,\di [X]_t-i \tilde{Q}_j \psi_j \,\di X_t , \quad \psi_j(0)=\psi_{j0},
\end{equation}
and $\phi_j$ is its corresponding noiseless target state. 
\end{lemma}

\begin{proof}
Squaring  $S_n$, we find
\begin{equation}
    S_n^2=\sum_{j=1}^n Q_j^2+\sum_{j=1}^n\sum_{k\neq j}^n Q_{jk},\quad Q_{jk}=I^{\otimes (j-1)}\otimes \tilde{Q}_j \otimes I^{\otimes (k-j-1)} \otimes \tilde{Q}_k \otimes I^{\otimes (n-k)}.
\end{equation}
By induction we prove the $\psi_{(n)}=\bigotimes_{j=1}^n \psi_j$. 

\smallskip

This holds trivially for $n=1$. Now, assume the statement holds for $n$. Then for $n+1$, we use It\^o's formula to obtain
\begin{equation}
\begin{aligned}
    \di\bigotimes_{j=1}^{n+1} \psi_j
    =-iH_{n+1}\bigotimes_{j=1}^{n+1} \psi_j\,\di t-iS_{n+1}\bigotimes_{j=1}^{n+1} \psi_j\,\di X_t-\frac{1}{2}S_{n+1}\bigotimes_{j=1}^{n+1} \psi_j \,\di [X]_t.
\end{aligned}
\end{equation}
As we have equal initial conditions, we indeed find $\psi_{(n+1)}=\bigotimes_{j=1}^{n+1} \psi_j$ for all $n\in\mathbb{N}$.

\smallskip

Analogously, we find $\phi_{(n)}=\bigotimes_{j=1}^{n} \phi_j$. For the fidelity, we then deduce
\begin{equation}
    |\phi_{(n)}^\dagger \psi_{(n)}|^2=\phi_{(n)}^\dagger \psi_{(n)}\psi_{(n)}^\dagger\phi_{(n)}=\left(\bigotimes_{j=1}^{n} \phi_j^\dagger \bigotimes_{j=1}^{n} \psi_j\right)\left(\bigotimes_{j=1}^{n} \psi_j^\dagger \bigotimes_{j=1}^{n} \phi_j\right)=\prod_{j=1}^n |\phi_j^\dagger\psi_j|^2. 
\end{equation}
\end{proof}

Note that these results do not only hold for qubits, but for any finite ensemble of finite dimensional quantum systems.

\section{Stochastic Integration}
\label{app:stochint}
Numerical verification of the analytic results is performed using stochastic integration. To solve for the noise and state simultaneously, we define $\Y:=(\psi,X)$. For the (OU) process, this gives the differential equation
\begin{equation}
    \di \Y=a(\Y)\,\di t+b(\Y)\,\di W_t,\quad a(\Y)=
\begin{pmatrix}
-iH+ikXS-\frac{\gamma^2}{2}S^\dagger S
  & \rvline & \mathbf{0} \\
\hline
  \mathbf{0}^\text{T} & \rvline &
-k
\end{pmatrix}\Y,\quad b(\Y)=\begin{pmatrix}
-i\gamma S
  & \rvline & \mathbf{0} \\
\hline
  \mathbf{0}^\text{T} & \rvline &
\frac{1}{X}
\end{pmatrix}\Y.
\end{equation}
These equations can be solved discretely over time steps $\Delta t$ using a numerical integration scheme. One possible scheme is the explicit (weak) first-order Euler-Maruyama scheme \cite{stochasticintegration1} as
\begin{equation}
    \Y_{n+1}=\Y_n+a(\Y_n)\Delta t+b(\Y_n)\mathcal{N}\sqrt{\Delta t},
\end{equation}
where $\mathcal{N}$ is a standard normal distribution. Throughout this work, convergences issues persisted using stochastic integration schemes for (OU) noise at higher evolution times, likely due to the non-Lipschitz $1/X$ dependence \cite{platen} and possibly the non-Euclidean space in which the states live. These convergence issues are absent for white noise and always occur below fidelities of $F=0.95$, which is not a relevant regime for pragmatic quantum computing and therefore not a deliberating issue. We have found that these issues are mitigated (but not resolved) when using a higher-order scheme such as the explicit (weak) second-order scheme due to Platen \cite{stochasticintegration2,platen}, which is used throughout this work. This scheme is given by 
$$
\begin{aligned}
& \Y_{n+1}=\Y_n+\frac{1}{2}\big(a(\bar{\Upsilon})+a(\Y_n)\big) \Delta t+\frac{1}{4}\big(b\left(\bar{\Upsilon}^{+}\right)+b\left(\bar{\Upsilon}^{-}\right)+2 b(\Y_n)\big) \mathcal{N}\sqrt{\Delta t} +\frac{1}{4}\big(b(\tilde{\Upsilon}^{+})-b(\widetilde{\Upsilon}^{-})\big)\left(\mathcal{N}^2-1\right) \sqrt{\Delta t}, 
\end{aligned}
$$
with supporting values $\bar{\Upsilon}=\Y_n+a(\Y_n) \Delta t+b(\Y_n) \mathcal{N}\sqrt{\Delta t}$ and $\bar{\Upsilon}^{ \pm}=\Y_n+a(\Y_n) \Delta t \pm b(\Y_n) \sqrt{\Delta t}$.

\section{Second-order Pauli Approximation}
\label{app:secondorder}

For the second-order results as in Fig.~\ref{fig:orderofapprox},  the closed system for the vector $\V$
\begin{equation}
\begin{aligned}
    \V=\bigg[|\phi^\dagger\psi|^2,|\phi^\dagger S\psi|^2,X(t)(\phi^\dagger S\psi\psi^\dagger \phi-\phi^\dagger \psi\psi^\dagger S\phi), X^2(t)|\phi^\dagger\psi|^2,X^2(t)|\phi^\dagger S\psi|^2,X^3(t)(\phi^\dagger S\psi\psi^\dagger \phi-\phi^\dagger \psi\psi^\dagger S\phi)\bigg],
\end{aligned}
\end{equation}
is found to be
\begin{equation}
\dot{\V}=\left[
\begin{array}{cccccc}
 -\gamma^2 & \gamma^2 & ik & 0 & 0 & 0 \\
 \gamma^2 & - \gamma^2 & -ik & 0 & 0 & 0 \\
 -2i\gamma^2 & 2i\gamma^2 & -(k+2 \gamma^2) & 2ik & -2ik & 0 \\
 \gamma^2 & 0 & -2 i\gamma^2 & -(2k+\gamma^2) & \gamma^2 & i k \\
 0 & \gamma^2 & 2 i\gamma^2 & \gamma^2 & -(2k+\gamma^2) &  -i k \\
 0 & 0 & 2\gamma^2 & 2i(k\mathbb{E}[X_t^2]-2\gamma^2) & -2i(k\mathbb{E}[X_t^2]-2\gamma^2) & -(3k+2 \gamma^2) \\
\end{array}
\right]\V, \quad \V(0)=\left[
\begin{array}{c}
 1  \\
 \S_0^2 \\
 0\\
 0 \\
 0\\
 0 
\end{array}
\right],
\end{equation}
with $\S_0=\phi_0^\dagger S\phi_0$. This system is solved numerically to retrieve the expectation value of the fidelity.

\section{Non-commuting}
For the non-commuting system we consider $H=\alpha \sigma_1$, $S=\sigma_2$, giving $[H,S]=2i\alpha\sigma_3$, where $\{\sigma_1,\sigma_2,\sigma_3\}$ can be any cyclic permutation of $\{\sigma_X,\sigma_Y,\sigma_Z\}$. The group structure of the Pauli matrices allows us to close the system of ODEs for the vector
\label{app:noncommute}
\begin{equation}
\begin{aligned}
\V = \begin{bmatrix}
	\F \\ |\phi^\dag \sigma_1 \psi|^2 \\ |\phi^\dag \sigma_2 \psi|^2 \\ |\phi^\dag \sigma_3 \psi|^2 \\ i(\phi^\dag \sigma_1 \psi \psi^\dag \phi - \phi^\dag \psi \psi^\dag \sigma_1 \phi) \\ i(\phi^\dag \sigma_2 \psi \psi^\dag \phi - \phi^\dag \psi \psi^\dag \sigma_2 \phi) \\ i(\phi^\dag \sigma_3 \psi \psi^\dag \phi - \phi^\dag \psi \psi^\dag \sigma_3 \phi) \\ \phi^\dag \sigma_2 \psi \psi^\dag \sigma_1 \phi + \phi^\dag \sigma_1 \psi \psi^\dag \sigma_2 \phi \\ \phi^\dag \sigma_3 \psi \psi^\dag \sigma_1 \phi + \phi^\dag \sigma_1 \psi \psi^\dag \sigma_3 \phi \\ \phi^\dag \sigma_3 \psi \psi^\dag \sigma_2 \phi + \phi^\dag \sigma_2 \psi \psi^\dag \sigma_3 \phi
\end{bmatrix}\in\mathbb{R}^{10}.
\end{aligned}
\end{equation}
We find that the matrix $A$ splits in a commutator part $A_c$ and a noise part $B^2$, resulting in the system
\begin{equation}
    \di \V=\alpha A_c \V\,\di t+\frac{1}{2}\gamma^2 B^2\V\,\di t+ B\V\,\di X_t,
\end{equation}
with
\begin{equation}
A_c = 
    \begin{pmatrix}
         0 & 0 & 0 &  0 & 0 & 0 & 0 & 0 & 0 & 0 \\
         0 & 0 & 0 &  0 & 0 & 0 & 0 & 0 & 0 & 0 \\
         0 & 0 & 0 &  0 & 0 & 0 & 0 & 0 & 0 & -2 \\
         0 & 0 & 0 &  0 & 0 & 0 & 0 & 0 & 0 & 2 \\
         0 & 0 & 0 &  0 & 0 & 0 & 0 & 0 & 0 & 0 \\
         0 & 0 & 0 &  0 & 0 & 0 & -2 & 0 & 0 & 0 \\
         0 & 0 & 0 &  0 & 0 & 2 & 0 & 0 & 0 & 0 \\
         0 & 0 & 0 &  0 & 0 & 0 & 0 & 0 & -2 & 0 \\
         0 & 0 & 0 &  0 & 0 & 0 & 0 & 2 & 0 & 0 \\
         0 & 0 & 4 & -4 & 0 & 0 & 0 & 0 & 0 & 0 \\
    \end{pmatrix},\qquad
B =
    \begin{pmatrix}
        0 & 0 & 0 & 0 & 0 & -1 & 0 & 0 & 0 & 0 \\
         0 & 0 & 0 & 0 & 0 & 0 & 0 & 0 & 1 & 0 \\
         0 & 0 & 0 & 0 & 0 & 1 & 0 & 0 & 0 & 0 \\
         0 & 0 & 0 & 0 & 0 & 0 & 0 & 0 & -1 & 0 \\
         0 & 0 & 0 & 0 & 0 & 0 & 1 & -1 & 0 & 0 \\
         2 & 0 & -2 & 0 & 0 & 0 & 0 & 0 & 0 & 0 \\
         0 & 0 & 0 & 0 & -1 & 0 & 0 & 0 & 0 & -1 \\
         0 & 0 & 0 & 0 & 1 & 0 & 0 & 0 & 0 & 1 \\
         0 & -2 & 0 & 2 & 0 & 0 & 0 & 0 & 0 & 0 \\
         0 & 0 & 0 & 0 & 0 & 0 & 1 & -1 & 0 & 0 \\
    \end{pmatrix}.
\end{equation}
Solving this system as in Sec.~\ref{sec:odes} is not possible since $[A_c,B]\ne 0$. However, a perturbation technique can be used when the Hamiltonian strength is much larger than the noise, i.e.\ $\varepsilon^2 :=\gamma^2/\alpha\ll 1$, which holds for many realistic systems. Rescaling time $\widehat t = t/\alpha$, we obtain a rescaled system of equations for $\widehat\V= \V_{t/\alpha}$ and $\widehat X = X_{t/\alpha}$,
\begin{equation}
    \di \widehat\V= A_c \widehat\V\,\di \widehat t + \frac{1}{2}\varepsilon^2 B^2\widehat\V\,\di \widehat t + B\widehat \V\,\di \widehat X_t,
\end{equation}
where $\widehat X$ has the quadratic variation $[\widehat X]_{\widehat t} = \varepsilon^2 \widehat t$. \\

In the following, we assume w.l.o.g.\ that $\alpha=1$ and $\gamma^2\ll 1$. Otherwise, we simply rescale time as above and write $\V$ in place of $\widehat\V$. Setting $\U:=\exp(-A_c t)\V$ to get
\begin{equation}
\label{eq:gho}
    \di\U=\frac{1}{2}\gamma^2 D^2(t)\U\, \di t+D(t)\U\, \di X_t,
\end{equation}
with 
\begin{equation}
D(t)=\exp(-A_c t)B\exp(A_c t)=\cos(2 t) B-\frac{1}{2}\sin(2 t)[A_c,B],
\end{equation}
where the second equality can be established according to a generalized harmonic oscillator \cite{introquantcalc}. Although \eqref{eq:gho} cannot be solved explicitly, an approximation to its solution can be found via the stochastic Magnus expansion \cite{stochasticmagnus}, which states that 
\begin{equation}
    \U=\exp(\L)=\exp\left(\sum_{n,r} \gamma^{2n-r}\L^{(r,n-r)}\right),
\end{equation}
where $\L^{(r,n-r)}$ are the expansion terms up to order $\gamma^2$ in $\L$. For general processes we find
\begin{equation}
    \U_t=\exp\left(\int_{0}^t D(s)\,\di X_s+O(\gamma^3)\right),
\end{equation}
which, for white noise, (WN) gives 
\begin{equation}
    \mathbb{E}[\U_t]\approx\exp\left(\frac{\gamma^2}{2}\int_{0}^t D^2(s)\,ds\right).
\end{equation}
For the fidelity under white noise, we obtain the approximation
\begin{equation}
    \F_t=\mathbb{E}[\U_t]\approx\frac{1}{2}+\frac{1}{2} \\C_1^2 e^{-2 \gamma^2 t}+\frac{1}{2} e^{-\gamma^2 t} \Bigl(\sinh (u) ((\C_2^2-\C_3^2) \cos (2a t)-2 \C_2 \C_3
   \sin (2a t))+\left(\C_2^2+\C_3^2\right) \cosh (u)\Bigr).
\end{equation}
For the (OU) process with $k>0$, the exponential can not be expressed analytically. Instead, one could approximate it by expanding the exponential up to the second order to find

\begin{align*}
    \mathbb{E}[\U_t] &\approx I+\frac{\gamma^2}{2}\int_{0}^t  D^2(s) ds \\
    &\qquad - \frac{\gamma^2}{2} k\int_0^t \left\{e^{-ks} D(s),\int_{0}^{s}e^{ks'}D(s') ds' \right\}ds+\frac{\gamma^2}{4}k\int_{0}^t \left\{e^{-2ks}D(s),\int_{0}^s e^{2ks'}D(s')ds'\right\}ds.
\end{align*}
\end{document}